\InputIfFileExists{public-build.cfg}{}{}

\documentclass[11pt]{article}

\usepackage[letterpaper,margin=1in]{geometry}
\usepackage{amsmath,amssymb,amsthm,mathtools}
\usepackage{microtype}
\usepackage{enumitem}
\usepackage{xcolor}
\usepackage{booktabs}
\usepackage{array}
\usepackage{float}
\usepackage{aliascnt}
\usepackage[numbers,sort&compress]{natbib}
\usepackage[colorlinks=true,linkcolor=blue!55!black,citecolor=blue!55!black,
  urlcolor=blue!55!black]{hyperref}
\usepackage[nameinlink,noabbrev]{cleveref}

\newcommand{\PaperAuthorBlock}{}
\newcommand{\PaperPdfAuthors}{Anonymous author}
\newif\ifpublicpaper
\ifdefined\BuildPublicPaper
  \publicpapertrue

\renewcommand{\PaperPdfAuthors}{Jack Fitzsimons}
\renewcommand{\PaperAuthorBlock}{%
  Jack Fitzsimons\\
  Oblivious%
}
{Public build requested without public-metadata.tex}{%
      Copy public-metadata.example.tex to public-metadata.tex and replace all placeholders.}%
  }
\else
  \publicpaperfalse
\fi

\hypersetup{
  pdftitle={Pure-DP Statistical Query Release at the Conjectured Square-Root Rate},
  pdfauthor={\PaperPdfAuthors},
  pdfsubject={Pure-DP statistical-query release at the conjectured ell-infinity rate},
  pdfkeywords={differential privacy, query release, private multiplicative weights, Rényi divergence, Maurey method}
}

\newtheorem{theorem}{Theorem}[section]
\newaliascnt{lemma}{theorem}
\newtheorem{lemma}[lemma]{Lemma}
\aliascntresetthe{lemma}
\newaliascnt{proposition}{theorem}
\newtheorem{proposition}[proposition]{Proposition}
\aliascntresetthe{proposition}
\newaliascnt{corollary}{theorem}

\aliascntresetthe{corollary}
\theoremstyle{definition}
\newaliascnt{definition}{theorem}

\aliascntresetthe{definition}
\newaliascnt{remark}{theorem}
\newtheorem{remark}[remark]{Remark}
\aliascntresetthe{remark}

\AddToHook{env/theorem/begin}{\crefalias{section}{theorem}}
\AddToHook{env/lemma/begin}{\crefalias{section}{lemma}}
\AddToHook{env/proposition/begin}{\crefalias{section}{proposition}}
\AddToHook{env/corollary/begin}{\crefalias{section}{corollary}}
\AddToHook{env/definition/begin}{\crefalias{section}{definition}}
\AddToHook{env/remark/begin}{\crefalias{section}{remark}}

\crefname{theorem}{Theorem}{Theorems}
\crefname{lemma}{Lemma}{Lemmas}
\crefname{proposition}{Proposition}{Propositions}
\crefname{corollary}{Corollary}{Corollaries}
\crefname{definition}{Definition}{Definitions}
\crefname{remark}{Remark}{Remarks}

\newcommand{\R}{\mathbb R}
\newcommand{\E}{\mathbb E}
\newcommand{\Prb}{\mathbb P}
\newcommand{\eps}{\varepsilon}
\newcommand{\cD}{\mathcal D}
\newcommand{\cQ}{\mathcal Q}
\newcommand{\cS}{\mathcal S}
\newcommand{\cV}{\mathcal V}
\newcommand{\Om}{\Omega}
\newcommand{\KL}{\mathrm{KL}}
\newcommand{\Ren}{D}
\newcommand{\Ans}{\operatorname{Ans}}
\newcommand{\defeq}{\mathrel{:=}}
\newcommand{\norm}[1]{\left\lVert#1\right\rVert}
\newcommand{\abs}[1]{\left\lvert#1\right\rvert}

\makeatletter
\renewenvironment{abstract}
  {\begin{center}\bfseries\abstractname\end{center}\begin{quotation}}
  {\end{quotation}}
\makeatother

\title{Pure-DP Statistical Query Release\\
at the Conjectured Square-Root Rate}
\author{\PaperAuthorBlock}
\date{}

\begin{document}
\maketitle

\begin{abstract}
Nikolov and Ullman posed two open problems about releasing $k$ statistical
queries on a universe of size $T$ under pure differential privacy
\cite{nikolov2021openproblem}.  Nikolov subsequently resolved their
normalized-Euclidean problem \cite{nikolov2023jl}.  We prove the remaining
conjectured upper bound.  For
every database size $n$ and privacy parameter $\eps>0$, there is an
$\eps$-differentially private mechanism with expected worst-coordinate error
at most
\[
  129e\min\left\{1,
  \sqrt{\frac{\log(2T)\log(2k)}{\eps n}}\right\}.
\]
The same envelope framework also gives an alternative, information-theoretic
mechanism with expected
normalized Euclidean error at most
\[
  \frac1{\sqrt{k}}\E\norm{M(x)-F(x)}_2
  \le 62e\min\left\{1,
  \sqrt{\frac{\log(2T)}{\eps n}}\right\}.
\]
This second bound has no dependence on the number of queries and recovers the
rate already achieved efficiently by Nikolov.  The shifted logarithms and
outer minima make both statements valid without additional parameter
assumptions.

Both constructions start from private multiplicative-weights (PMW)
transcripts.  For worst-coordinate error, each round privately selects a
signed query.  For normalized Euclidean error, the query is sampled uniformly
and independently of the data, and only its sign is selected privately.
We replace either transcript probability mass function by a
distance-penalized likelihood envelope: a transcript receives the largest
likelihood it has under any database, discounted exponentially by that
database's Hamming distance from the true input.  After normalization, this
yields the pointwise likelihood-ratio bound required for pure DP.  To prove
that the modification preserves
accuracy, a likelihood-level Maurey argument upper-bounds each Hamming-ball
maximum by a small family of auxiliary PMW laws.  R\'enyi moment bounds
control nearby balls, a direct mixture bound controls distant balls, and
grouping radii at scale $1/\eps$ prevents an additional $1/\eps$ loss.

The mechanisms are information-theoretic.  A companion Lean 4 development
machine-checks both finite constructions, pure privacy after deterministic
decoding, and both displayed all-regimes upper bounds.
\end{abstract}
\clearpage

\section{Introduction}

Suppose a curator wants to publish many averages of the same private
database.  For a single average, sensitivity is $O(1/n)$, so
sensitivity-calibrated noise is small.  Answering $k$ averages independently,
however, incurs cumulative privacy cost.  The query-release problem asks
whether the curator can exploit the fact that all answers come from one
empirical distribution.

Formally, let $\cD$ be a finite data universe, let $\cQ$ be a nonempty finite
family of bounded functions $q:\cD\to[-1,1]$, and let
$x=(x_1,\ldots,x_n)\in\cD^n$.  The answer to $q$ is
\[
  F_q(x)=\frac1n\sum_{i=1}^n q(x_i),
  \qquad F(x)=(F_q(x))_{q\in\cQ}.
\]
We seek one private random vector that approximates every coordinate of
$F(x)$.  This is the statistical-query release problem in differential privacy
\cite{dwork2006calibrating,dwork2014foundations}.

The gap addressed here is between approximate and pure privacy.  Approximate
DP allows an additive $\delta>0$ term in the privacy inequality, whereas pure
DP sets $\delta=0$.  Private
multiplicative weights and iterative database constructions achieve
$n^{-1/2}$-type accuracy under approximate DP
\cite{hardt2010multiplicative,gupta2012iterative,hardt2012simple}.  The
classical pure-DP small-database method gives only a cube-root upper bound for
worst-coordinate error \cite{blum2013learning}.  Nikolov and Ullman isolated
this gap as Open Problem~1 on DifferentialPrivacy.org.  Their statement
explicitly focuses on expected error, while noting that query-release
algorithms often also provide high-probability guarantees
\cite{nikolov2021openproblem}.  In precisely this expected-$\ell_\infty$
sense, they conjectured the pure-DP rate
\[
  \sqrt{\frac{\log k\log T}{\eps n}},
\]
and noted that it would match the known lower-bound scaling in the standard
high-dimensional regimes \cite{hardt2011thesis,nikolov2021openproblem}.
Open Problem~2 asks whether expected $\ell_2$ error $\alpha\sqrt{k}$ can be
achieved with
\[
  \alpha=\sqrt{\frac{\log T}{\eps n}},
\]
with no dependence on $k$.  Nikolov resolved that problem with an efficient
Johnson--Lindenstrauss mechanism \cite{nikolov2023jl}.  We establish the
conjectured Open Problem~1 upper bound, uniformly over all finite parameter
choices.  For completeness, \cref{sec:l2} shows that the envelope method also
recovers the known Open Problem~2 rate by a different, information-theoretic
mechanism.  We make no priority claim for that second rate.  The corresponding
lower bounds apply in their stated nondegenerate high-dimensional regimes.

The obstacle is a mismatch between the privacy information supplied by PMW
and the privacy guarantee we need.  At the square-root parameters, a
selection-only PMW transcript is accurate and its privacy loss has good
finite moments.  Such moment bounds underlie concentrated and approximate
DP, but pure DP requires a pointwise likelihood-ratio inequality for every
transcript.  Conditioning on a likely ``good'' privacy-loss event does not
repair this mismatch: two neighboring datasets could retain different
supports.  Nor can we union-bound over nearby databases, because even a small
Hamming ball may contain exponentially many of them.

We repair the transcript distribution itself.  Let $P_x$ be the distribution
of the selection-only PMW transcript on database $x$, and let $p_x(\omega)$
be the probability of transcript $\omega$.  Before releasing anything, form
the unnormalized envelope
\[
  \widetilde p_x(\omega)
  =\max_{y\in\cD^n}e^{-(\eps/2)d_H(x,y)}p_y(\omega).
  \tag{1.1}\label{eq:intro-envelope}
\]
For a fixed transcript, the expression asks how likely that transcript could
be under any database, but charges a factor $e^{-\eps/2}$ for each changed
row.  Moving the input $x$ by one row changes every discounted likelihood by
at most $e^{\eps/2}$.  The same is true of their maximum.  After dividing by
the total envelope mass, the changing normalizer costs a second
$e^{\eps/2}$, and the resulting transcript mechanism is pure $\eps$-DP.  This
factor-two normalization pattern mirrors the metric-extension geometry of
Borgs, Chayes, Smith, and Zadik \cite{borgs2018private}.

There is also a direct differential-privacy antecedent.  If each reference
law $P_y$ is a point mass at a deterministic statistic $f(y)$, then for every
$\omega$ in the range of $f$, \eqref{eq:intro-envelope} becomes
\[
  \widetilde p_x(\omega)
  =\exp\!\left(-\frac{\eps}{2}
    \min_{y:f(y)=\omega} d_H(x,y)\right),
\]
which is exactly the discrete inverse-sensitivity mechanism studied by Asi
and Duchi after normalization \cite{asi2020near}.  Here the reference laws are
instead nondegenerate PMW transcript distributions, and the central utility
question is whether their likelihood envelope remains close enough to the
original PMW law.

Privacy follows directly from the metric form of the envelope; utility is the
main technical issue.  Because the envelope pointwise dominates the original
PMW mass function, it may assign additional probability to
transcripts that decode to inaccurate answers.  Utility therefore reduces to
two questions: how much total probability mass the envelope adds, and how
much error-weighted mass it adds.  The rest of the proof answers both questions
without enumerating the databases in a Hamming ball.

\subsection{Main theorem}

Throughout, databases have fixed size and two databases are adjacent when one
row is replaced, equivalently when their Hamming distance is one.  A
randomized mechanism is pure $\eps$-differentially private if, for neighboring
$x,x'$ and every subset $E$ of its output space,
\[
  \Prb[M(x)\in E]\le e^\eps\Prb[M(x')\in E].
  \tag{1.2}\label{eq:event-dp}
\]

\begin{theorem}[The conjectured pure-DP query-release upper bound, with explicit constant]
\label{thm:main}
Let $\cD$ be finite with $T=\abs{\cD}\ge1$, let $\cQ$ be nonempty and finite with
$k=\abs{\cQ}\ge1$, let $n\ge1$, and suppose every $q\in\cQ$ maps $\cD$ to
$[-1,1]$.  For every $\eps>0$, there exists a randomized mechanism $M$ with
output in $\R^\cQ$ that is pure $\eps$-differentially private and satisfies
\[
  \sup_{x\in\cD^n}
  \E\norm{M(x)-F(x)}_\infty
  \le 129e\min\left\{1,
  \sqrt{\frac{\log(2T)\log(2k)}{\eps n}}\right\}.
  \tag{1.3}\label{eq:main-rate}
\]
\end{theorem}

\begin{theorem}[Envelope recovery of the normalized Euclidean rate]
\label{thm:main-l2}
Under the hypotheses of \cref{thm:main}, there exists a randomized mechanism
$M_2$ with output in $\R^\cQ$ that is pure $\eps$-differentially private and
satisfies
\[
  \sup_{x\in\cD^n}
  \frac1{\sqrt{k}}\E\norm{M_2(x)-F(x)}_2
  \le62e\min\left\{1,
  \sqrt{\frac{\log(2T)}{\eps n}}\right\}.
  \tag{1.4}\label{eq:main-l2-rate}
\]
\end{theorem}

Both mechanisms constructed here have finite output support.  They are information-theoretic:
the transcript spaces and the maximum in \eqref{eq:intro-envelope} are finite,
but they may be exponentially large.

\subsection{Construction and proof strategy}
\label{sec:overview}

Each mechanism consists of three steps.  They share the envelope and decoder;
their base transcript laws differ only in how the signed query is sampled.
Likelihood domination bounds the utility loss introduced by the second step.

\begin{enumerate}[leftmargin=2.2em]
  \item \emph{Define the PMW transcript distribution.}
  Starting from the uniform synthetic distribution on $\cD$, update in the
  direction of one signed query per round.  For \cref{thm:main}, an
  exponential mechanism selects the signed query with largest discrepancy.
  For \cref{thm:main-l2}, the ordinary query is public uniform randomness and
  a two-point exponential mechanism selects only its sign.  Either transcript
  records query indices and signs, but no numerical answers.

  \item \emph{Replace the transcript distribution by its privacy envelope.}
  For every possible transcript, compute \eqref{eq:intro-envelope}, sum these
  weights to obtain a normalizer $Z_x$, and sample the transcript from
  $\widetilde p_x/Z_x$.  This is the transcript sampled by the private core of
  the mechanism.

  \item \emph{Decode publicly.}  The signed queries determine every PMW
  update, so the entire synthetic trajectory can be replayed from the
  transcript.  Output the query answers of the average synthetic
  distribution.  This deterministic decoding does not spend privacy.
\end{enumerate}

The first step supplies utility, but its direct pure-DP bound exceeds the
target privacy budget.  The second step supplies pure privacy but can distort
the base distribution.  The proof is therefore about showing that the
distortion is small where the decoder's error is large.

\paragraph{The utility quantity.}
Fix the true database $x$, and let $\ell_x(\omega)$ denote the worst-coordinate
error of the answer decoded from transcript $\omega$.  Since the normalized
mechanism samples with probability $\widetilde p_x(\omega)/Z_x$, its expected
error is
\[
  \frac{1}{Z_x}\sum_\omega
    \widetilde p_x(\omega)\ell_x(\omega).
\]
The pointwise majorant contains the accurate base mass $p_x$, and this also
shows $Z_x\ge1$.  Normalization can nevertheless shift probability toward
inaccurate transcripts.  It remains to bound the extra unnormalized mass and
extra error-weighted mass contributed by databases $y\ne x$.

\paragraph{Likelihood-level domination for a Hamming ball.}
Consider databases within Hamming distance $R$ of $x$.  Their query-answer
vectors differ from that of $x$ by averages of at most $R$ elementary row
replacements.  For the utility analysis, a likelihood-level Maurey argument
upper-bounds the maximum transcript likelihood over every database in this
ball by the maximum over a small family of auxiliary PMW laws.  The family is
obtained by sampling $m$ row replacements and has at most $(T^2+1)^m$
descriptions.  Although the sampling law for a particular database depends on
that database, its support family depends only on $x$, $R$, and $m$.  A
fixed-transcript log-sum-exp and Hoeffding estimate gives the pointwise
domination, at multiplicative cost
\[
  \exp\!\left\{\frac{\kappa R^2}{m}\right\},
  \qquad \kappa=\frac{2J\eta^2}{n^2}.
\]
Here $J$ is the number of PMW rounds and $\eta$ is the exponential-mechanism
selection strength.
The auxiliary answer vectors need not correspond to real databases, and a
maximizing vector may depend on the fixed transcript.  This is why the PMW
analysis is first extended to antisymmetric surrogate targets; they are
analytical interpolation points, never possible outputs or private inputs.

\paragraph{Why there are two distance regimes.}
After this domination, two different estimates are available.  For a nearby
ball, finite-order Rényi bounds control the maximum likelihood over the
codebook, and Hölder's inequality transfers the base PMW $L_2$ error bound.
For a distant ball, this moment estimate becomes expensive.  There it is
better to replace the maximum by the sum of the codebook likelihoods; each is
a probability distribution of total mass one.  Optimizing the codebook size
in either regime gives ball growth
\[
  \exp\{O(R\sqrt{\kappa\log(T^2+1)})\}.
\]
The PMW parameters are chosen so that this growth is much slower than the
envelope penalty $e^{-\eps R/2}$.

Here is the quantitative interface proved later.  Define the radius-$R$ ball
maximum
\[
  H_R(\omega)=\max_{y:d_H(x,y)\le R}p_y(\omega),
  \qquad
  \kappa=\frac{2J\eta^2}{n^2},
  \qquad L=\log(T^2+1).
\]
The moment and mixture regimes together give
\[
 \sum_\omega H_R(\omega)\le e^{5R\sqrt{\kappa L}},
 \qquad
 \sum_\omega H_R(\omega)\ell_x(\omega)
 \le e^{5R\sqrt{\kappa L}}
      \left(\alpha_0+\frac{4R}{n}\right).
\]

\paragraph{Why radii are blocked.}
Summing the resulting estimate separately over every integer radius would
still introduce an unwanted factor $1/\eps$.  We instead group radii into
blocks of width $\Theta(1/\eps)$.  Within a block, the outer radius bounds the
likelihood growth while the inner radius supplies the privacy penalty.  The
block weights then have constant total mass and constant block-index-weighted
mass.  Concretely, for
$B=\max\{1,\lceil3/\eps\rceil\}$ and $R_j=(j+1)B$, the envelope satisfies
the pointwise block bound
\[
  \widetilde p_x(\omega)
  \le p_x(\omega)+\sum_{j=0}^{n-1}
  e^{-(\eps/2)(jB+1)}H_{R_j}(\omega).
\]
Since a block has width $\Theta(1/\eps)$, multiplying its radius by the
$1/n$ contribution to error gives $O(1/(\eps n))$.  The resulting explicit
bounds are
\[
  1\le Z_x\le1+2e,\qquad
  \sum_\omega\widetilde p_x(\omega)\ell_x(\omega)
    \le(1+2e)\alpha_0+\frac{144e}{\eps n},
\]
where $\alpha_0$ is the base PMW error.  Balancing the PMW round, selection,
and update parameters makes
$\alpha_0=O(\sqrt{\log T\log k/(\eps n)})$ and absorbs the remaining linear
term.

Sections~\ref{sec:pmw}--\ref{sec:parameters} present these components and their
analysis in order: base transcript, privacy envelope, Hamming-ball likelihood
domination, near/far estimates, blocked summation, and parameter choice.  The
explicit parameter verification is isolated in \cref{sec:parameters}; the
preceding sections establish the mechanism and structural estimates.
Section~\ref{sec:l2} then gives the uniform-query modification and the proof
of \cref{thm:main-l2}; it reuses the envelope analysis and records only the
places where the utility and likelihood estimates change.

\subsection{Context and comparison with prior work}
\label{sec:related}

The closest results differ in privacy notion, error criterion, or
computational model.

\begin{center}
\begin{tabular}{@{}
  >{\raggedright\arraybackslash}p{0.22\linewidth}
  >{\raggedright\arraybackslash}p{0.31\linewidth}
  >{\raggedright\arraybackslash}p{0.37\linewidth}@{}}
\toprule
Line of work & Guarantee & Relation to this result \\
\midrule
Prior pure-DP upper bounds
  & $O((\log T\log k/(\eps n))^{1/3})$ worst-coordinate error
  & Net mechanisms, pure PMW, and MWEM retain exponent $1/3$
    \cite{blum2013learning,hardt2011thesis,hardt2012simple}. \\
Known lower bounds
  & Square-root worst-coordinate error in nondegenerate regimes
  & These lower bounds are stated under explicit high-dimensional parameter
    conditions
    \cite{hardt2011thesis,lyu2025fingerprinting}. \\
Approximate-DP PMW/MWEM
  & Square-root-type dependence on $1/n$, with $\delta$-dependent factors
  & These guarantees permit $\delta>0$; MWEM also has the separate pure
    cube-root analysis
    \cite{hardt2010multiplicative,gupta2012iterative,hardt2012simple}. \\
Metric extension
  & Extends an already pure metric-DP mechanism at factor-two privacy cost
  & Its inherited utility guarantee is restricted to the starting subset
    \cite{borgs2018private}. \\
Nikolov's JL mechanism
  & The conjectured normalized $\ell_2$ rate under pure DP, independent of $k$
  & Efficiently resolves Open Problem~2 \cite{nikolov2023jl}. \\
This work
  & The conjectured expected-$\ell_\infty$ rate under pure DP
  & Resolves Open Problem~1; also gives an alternative, formally verified
    envelope proof of the known normalized-$\ell_2$ rate. \\
\bottomrule
\end{tabular}
\end{center}

Known generic pure-DP mechanisms have cube-root worst-coordinate
error: the net mechanism, interactive pure PMW, and MWEM give
$O((\log T\log k/(\eps n))^{1/3})$, up to their stated probability terms
\cite{blum2013learning,hardt2011thesis,hardt2012simple}.  Convex-geometric,
factorization, Johnson--Lindenstrauss, and instance-optimal mechanisms give
strong workload-sensitive or average-error guarantees
\cite{nikolov2013geometry,edmondsEtAl2020factorization,
nikolov2023jl,blasiok2019instance}.  PREM gives relative-error synthetic
release under pure and approximate DP; that work explicitly relates its
purely additive gap to the universal pure-DP problem
\cite{ghaziEtAl2025prem}.  Nikolov, Tang, and Ullman's online factorization
theorem treats nonadaptively chosen query streams using a Gaussian
$(\eps,\delta)$-DP mechanism, while their adaptive small-dataset result also
uses approximate DP \cite{nikolovEtAl2026online}.  These results address
different error objectives, privacy notions, or interaction models from the
offline pure-DP guarantee considered here.

For normalized Euclidean error, the general pure-DP upper bound recorded when
the open problem was posed was
$\alpha\lesssim(\log^2 k\,\log^{3/2}T/(\eps n))^{1/2}$
\cite{nikolov2013geometry,nikolov2021openproblem}.  Nikolov subsequently
replaced it by $(\log T/(\eps n))^{1/2}$, with no dependence on $k$, using an
efficient Johnson--Lindenstrauss mechanism \cite{nikolov2023jl}.
\Cref{thm:main-l2} recovers this rate through the transcript-envelope route.

A contemporaneous study by Ghazi, Guzm\'an, Kamath, Knop, Kumar, and
Manurangsi gives fixed-parameter tractable synthetic-data algorithms when the
query-incidence graph has bounded treewidth.  Its general pure-DP guarantee
retains the cube-root sparse-regime term and it separately identifies the
optimal universal pure-DP rate as open, so its computational and structural
contribution is complementary to the rate proved here
\cite{ghaziEtAl2026fpt}.

For worst-coordinate error, PMW, boosting, and iterative database
constructions explain why $n^{-1/2}$-type accuracy is possible once
approximate privacy is allowed
\cite{dwork2010boosting,hardt2010multiplicative,gupta2012iterative,hardt2012simple}.
Hardt's lower-bound proof is a discrete adaptation of the geometric packing
argument of Hardt and Talwar \cite{hardt2010geometry,hardt2011thesis}.  For
every fixed $c>0$, it constructs, for fixed $\eps$, sufficiently large $n$,
and $k>n^{1+c}$, a workload with error
\[
  \Omega\!\left(
  \sqrt{\frac{\log k\,\log(T/n)}{\eps n}}
  \right)
\]
with constant probability \cite{hardt2011thesis}.  Thus it is also an
expected-error lower bound up to a constant.  When
$T\ge n^{1+c'}$ for a fixed $c'>0$, as in the usual high-dimensional
specialization, $\log(T/n)=\Theta(\log T)$.  Building on the
fingerprinting-code lower-bound program of Bun, Ullman, and Vadhan
\cite{bun2014fingerprinting}, Lyu and Talwar more recently prove, under
explicit high-dimensional hypotheses, that target error $\alpha$ requires
\[
 n=\Omega\!\left(
 \frac{\sqrt{\log T\log(1/\delta)}\,\log k}{\eps\alpha^2}
 \right)
\]
for $\ell_\infty$ release \cite{lyu2025fingerprinting}.  A pure-DP mechanism
is also $(\eps,\delta)$-DP; choosing $\delta=2/T$ in the admissible regime
recovers $\alpha=\Omega(\sqrt{\log T\log k/(\eps n)})$.  These lower bounds
establish tight square-root dependence under their stated high-dimensional
assumptions.

Lin, Wang, Ma, and Wang recently gave a randomized-postprocessing framework
that, under stated conditions, purifies approximate DP
\cite{lin2025purifying}.  Their query-release instantiation runs MWEM,
subsamples the resulting synthetic database to control the output-space cost
of purification, and then applies a binary encoding; for
$\cD=\{0,1\}^d$, its guarantee is
$\widetilde{O}((d/(n\eps))^{1/3})$.  Thus this black-box route retains the
cube-root parameter balance.  The present mechanism is not a postprocessing
of an approximate-DP release.  Instead, it uses the exact convex form of a
fixed PMW transcript likelihood to prove the Maurey domination in
\cref{lem:likelihood-maurey}.  This is the step that turns local moment control
into a globally pure law without losing the desired rate.

The envelope uses the same factor-two normalization accounting as the
private-extension theorem of Borgs, Chayes, Smith, and Zadik
\cite{borgs2018private,borgsEtAl2018revealing}.  Applying that theorem with
$H=\cD^n$ would first require pointwise privacy of the base family; for proper
$H$, its inherited utility statement is confined to $H$.  For our PMW
reference family, a direct
pointwise privacy-loss bound costs $4J\eta/n$ and yields the cube-root
tradeoff.  At the square-root parameters we instead retain finite-order
likelihood-ratio moments and prove directly that the resulting envelope has
bounded normalization overhead and error-weighted mass.  The factor-two
geometry is inherited; the proof-specific ingredients that produce the new
rate are the likelihood-level Maurey domination and blocked utility analysis.

\section{Analytic tools}

Write
\[
  d_H(x,y)=\abs{\{i:x_i\ne y_i\}}
\]
for replacement Hamming distance on $\cD^n$.  For a probability distribution
$\mu$ on $\cD$ and a function $h:\cD\to\R$, write
$h(\mu)=\sum_{d\in\cD}\mu(d)h(d)$.  Thus $q(\mu)$ and $s(\mu)$ are ordinary
expectations of a query and a signed query.  For a probability mass $P$ on a
finite set, $\E_PG=\sum_\omega P(\omega)G(\omega)$ and
$\norm{G}_{L_r(P)}=(\E_P\abs{G}^r)^{1/r}$ for $1\le r<\infty$.  All
logarithms are natural.

All probability spaces below are finite, so probabilities and expectations
reduce to finite sums.  The analysis concerns how transcript likelihoods
change with the database and how many likelihoods must be controlled
simultaneously.

For probability masses $P,Q$ on a common finite support, the order-$r$ Rényi
divergence, $r>1$, is
\[
  \Ren_r(P\Vert Q)
  =\frac1{r-1}\log\sum_\omega
       Q(\omega)\left(\frac{P(\omega)}{Q(\omega)}\right)^r.
\]
We use Rényi divergence as a compact way to record moments of the likelihood
ratio $P/Q$, equivalently the exponentiated privacy loss; the final mechanism
is proved pure DP directly
\cite{mironov2017renyi,bun2016concentrated}.

We record four finite-space lemmas, stated and proved in
\cref{app:elementary-proofs}.
First, a log-likelihood ratio bounded in $[-a,a]$ has
$\Ren_r(P\Vert Q)\le ra^2/2$, and these bounds add over adaptive rounds
(\cref{lem:bounded-lr}).  This turns the sensitivity of one PMW selection into
a transcript moment bound.  Second, Hedge regret is at most
$\log(T)/\gamma+\gamma J/2$ (\cref{lem:hedge}); this is the deterministic
potential argument behind PMW utility.  Third, a softmax draw from $K$ scores
has gap tail $Ke^{-\eta u}$ and $L_2$ gap at most
$4\log(2K)/\eta$ (\cref{lem:softmax-gap}); this controls the error from
privately selecting a violated query.  Finally, if
$\Ren_a(P_i\Vert P_0)\le aA$, Hölder's inequality bounds the maximum of $N$
likelihoods by a factor
$\exp\{\log(N)/a+(a-1)A\}$ times the conjugate $L_b(P_0)$ norm, where
$b=a/(a-1)$
(\cref{lem:renyi-max}).  The bounded-ratio and Rényi-maximum lemmas record
the standard real-order route for context.  Sections~5--8 do not invoke
\cref{prop:transcript-renyi,lem:renyi-max}: the load-bearing proof instead
derives exactly the required integer moments in
\cref{lem:integer-likelihood-moment} and carries out the maximum-and-Hölder
calculation inline in \eqref{eq:near-max}.  This is also the route formalized
in Lean.

\section{The signed selection-only PMW transcript}
\label{sec:pmw}
\label{sec:base}

This section constructs the accurate base transcript previewed in
\cref{sec:overview}.  PMW maintains a synthetic distribution and
repeatedly searches for a query on which its answer disagrees with the private
database.  Some PMW presentations additionally release a noisy numerical
answer at each round.  Here the transcript omits these measurements: once a
signed query has been selected, a fixed Hedge update determines the next
synthetic distribution.  The finite transcript therefore contains only query
indices and signs.

This distinction matters.  MWEM selects an inaccurate query by the
exponential mechanism and separately releases a Laplace-noised numerical
answer; basic composition gives its pure-DP cube-root guarantee, while
advanced composition gives a separate approximate-DP analysis
\cite{hardt2012simple}.  Our selection rule is likewise an exponential
mechanism over signed queries \cite{mcsherry2007mechanism}, but it releases no
measurement.  The resulting finite transcript law is an analytic input to
the envelope, not itself the final private mechanism.

To encode both the inaccurate coordinate and the direction of its error, give
each query two signs.  It is useful to retain distinct query indices even when
two queries coincide as functions.  Define
\[
  \cS=\cQ\times\{-1,+1\},\qquad K=\abs{\cS}=2k.
\]
For $s=(q,\sigma)\in\cS$, write $s(d)=\sigma q(d)$ and
$-s=(q,-\sigma)$.  Thus $s:\cD\to[-1,1]$ and the involution has no fixed
point.  The positive copy is favored when the synthetic answer to $q$ is too
small; the negative copy is favored when it is too large.

An \emph{admissible target} is an antisymmetric signed-answer vector
$a=(a_s)_{s\in\cS}$ satisfying $a_{-s}=-a_s$.  A dataset $x$ induces the
target
\[
  a(x)_s=\frac1n\sum_{i=1}^n s(x_i).
  \tag{3.1}\label{eq:dataset-target}
\]
Admissible targets serve as analytical inputs to transcript distributions.
They need not lie in $[-1,1]^{\cS}$ or equal $a(y)$ for any database $y$;
they are neither private inputs nor released outputs.  Some will arise by
averaging row changes.  The PMW selection rule depends only on the target
coordinates, so it remains well-defined at these interpolation points.  To
measure their scale, define
\[
  d_*(a,b)=\frac n2\max_{s\in\cS}\abs{a_s-b_s}.
  \tag{3.2}\label{eq:target-distance}
\]
Then $d_*(a(x),a(y))\le d_H(x,y)$: replacing one row changes every ordinary
query answer by at most $2/n$, and the factor $n/2$ converts this answer-space
change back to Hamming units.

Fix a number of rounds $J\ge1$, selection parameter $\eta>0$, and update
parameter $0<\gamma\le1$.  Set $\mu_0$ to be uniform on $\cD$.  After
selecting $s_t$, update the synthetic distribution by
\[
  \mu_t(d)=\frac{\mu_{t-1}(d)e^{\gamma s_t(d)}}
  {\sum_{z\in\cD}\mu_{t-1}(z)e^{\gamma s_t(z)}}.
  \tag{3.3}\label{eq:hedge-update}
\]
At round $t$, the preceding selections have fixed $\mu_{t-1}$.  Sample
$s_t\in\cS$ according to
\[
  \Prb_a(s_t=s\mid s_{<t})
  =\frac{\exp\{\eta(a_s-s(\mu_{t-1}))\}}
  {\sum_{u\in\cS}\exp\{\eta(a_u-u(\mu_{t-1}))\}},
  \tag{3.4}\label{eq:selection-kernel}
\]
then apply \eqref{eq:hedge-update}.  The transcript
space is $\Om=\cS^J$.  Write $P_a$ for its law and $p_a(\omega)>0$ for its
mass function.

For one-round intuition, suppose $a=a(x)$ and
$F_q(x)>q(\mu_{t-1})$.  Then the positive copy $(q,+1)$ has positive
discrepancy and the update factor $e^{\gamma q(d)}$ shifts synthetic mass
toward records with larger $q$-value.  If the synthetic answer is too large,
the negative copy has the corresponding discrepancy and reverses the update.
Thus a signed selection specifies both what is inaccurate and how to correct
it.  Formula \eqref{eq:selection-kernel} is precisely an exponential-mechanism
soft maximum of these discrepancies.

The base answer is the query vector of the average pre-update synthetic
distribution,
\[
  \bar\mu(\omega)=\frac1J\sum_{t=0}^{J-1}\mu_t,
  \qquad
  \Ans(\omega)=\bigl(q(\bar\mu(\omega))\bigr)_{q\in\cQ}.
  \tag{3.5}\label{eq:base-answer}
\]
Averaging is the online-to-batch step: Hedge controls the average discrepancy
over the selected rounds, and linearity turns that control into a bound for
every query at $\bar\mu$.  Only the signed-query transcript is random; no
numerical measurements are released during these PMW rounds.

\subsection{Contextual real-order Rényi control}

We next quantify how the transcript distribution changes when its target
changes.  The quantity $\kappa$ below is the quadratic privacy-loss scale of all
$J$ selections.  When $a=a(x)$ and $b=a(y)$, the right-hand side grows with
the square of the Hamming distance.  We prove the stronger statement for all
admissible targets because the likelihood-domination argument will create
targets that are not datasets.

\begin{proposition}[Transcript Rényi bound]
\label{prop:transcript-renyi}
Define
\[
  \kappa=\frac{2J\eta^2}{n^2}.
  \tag{3.6}\label{eq:kappa-def}
\]
For all admissible targets $a,b$ and every $r>1$,
\[
  \Ren_r(P_a\Vert P_b)
  \le 4r\kappa d_*(a,b)^2.
  \tag{3.7}\label{eq:extended-renyi}
\]
\end{proposition}

\begin{proof}
Put $\Delta=\max_s\abs{a_s-b_s}=2d_*(a,b)/n$.  At a common preceding
transcript, every logit in \eqref{eq:selection-kernel} changes by at most
$\eta\Delta$.  The selected logit and the log normalizer therefore each
change by at most $\eta\Delta$.  The two conditional laws have pointwise log
likelihood ratio at most $2\eta\Delta=4\eta d_*(a,b)/n$ in absolute value.
Apply the adaptive part of \cref{lem:bounded-lr} for $J$ rounds and substitute
\eqref{eq:kappa-def}.
\end{proof}

\Cref{prop:transcript-renyi} records the usual arbitrary-real-order moment
bound for context.  It is not invoked in the proof of \cref{thm:main}, which
uses the exact integer-order statement in
\cref{lem:integer-likelihood-moment}.  Neither form provides pure privacy at
the target parameters: converting the moment control to a pointwise
likelihood-ratio bound would recover the linear composition cost.

\subsection{Second-moment utility and robustness}

The selected query is unlikely to be far below the largest current
discrepancy.  Hedge turns this per-round statement into accuracy of the
average synthetic distribution.  We prove a second-moment, rather than only
an expected-error, bound because the near-ball argument will later pair the
error with a likelihood ratio using Hölder's inequality with exponent at
most two.

Fix a dataset $x$ and an admissible target $a$.  Define
\[
  \Delta_x(a)=\max_{s\in\cS}\abs{a_s-a(x)_s}
  =\frac{2d_*(a,a(x))}{n},
\]
and, at round $t$,
\[
  D_t(s)=a_s-s(\mu_{t-1}),\quad
  D_t^\star=\max_{s\in\cS}D_t(s),\quad
  G_t=D_t^\star-D_t(s_t).
  \tag{3.8}\label{eq:gap-def}
\]
Here $D_t^\star$ is the largest current signed discrepancy and $G_t$ is the
selection gap: how far the softmax draw falls below that maximum.  Because
$D_t(-s)=-D_t(s)$, we have $D_t^\star\ge0$.

\begin{proposition}[Robust PMW utility]
\label{prop:robust-utility}
Let
\[
  \ell_x(\omega)=\norm{\Ans(\omega)-F(x)}_\infty,
  \qquad L_0=\log T,\qquad L_S=\log(4k).
\]
Then
\[
 \norm{\ell_x}_{L_2(P_a)}
 \le
 \frac{L_0}{\gamma J}+\frac\gamma2+
 \frac{4L_S}{\eta}+\frac{4d_*(a,a(x))}{n}.
 \tag{3.9}\label{eq:robust-l2}
\]
\end{proposition}

The first three terms in \eqref{eq:robust-l2} are respectively the PMW costs
for the initial Hedge potential, update discretization, and private query
selection.
The final term measures robustness to analyzing a surrogate target $a$ instead
of the true dataset target $a(x)$.

\begin{proof}
Conditionally on the past, \eqref{eq:selection-kernel} is a softmax on the
scores $D_t(s)$.  Since $\log(2K)=\log(4k)=L_S$,
\cref{lem:softmax-gap} and the tower property give
$\norm{G_t}_{L_2(P_a)}\le4L_S/\eta$.  Minkowski's inequality therefore gives
\[
  \left\lVert\frac1J\sum_{t=1}^JG_t\right\rVert_{L_2(P_a)}
  \le\frac{4L_S}{\eta}.
  \tag{3.10}\label{eq:average-gap}
\]

Use the empirical distribution of $x$ as the comparator in
\cref{lem:hedge}.  Pointwise in the transcript,
\[
\begin{aligned}
 \frac1J\sum_tD_t(s_t)
 &\le \frac{L_0}{\gamma J}+\frac\gamma2+\Delta_x(a),\\
 \frac1J\sum_tD_t^\star
 &\le \frac{L_0}{\gamma J}+\frac\gamma2+\Delta_x(a)
       +\frac1J\sum_tG_t.
\end{aligned}
\tag{3.11}\label{eq:average-max-score}
\]
For every $s\in\cS$,
\[
  a_s-s(\bar\mu)=\frac1J\sum_tD_t(s)
  \le\frac1J\sum_tD_t^\star.
\]
Apply this once to $s$ and once to $-s$ to bound the absolute difference
between $s(\bar\mu)$ and $a_s$.  A final triangle inequality from $a$ to
$a(x)$ contributes another $\Delta_x(a)$.  Take the $L_2(P_a)$ norm, use
\eqref{eq:average-gap}, and substitute
$2\Delta_x(a)=4d_*(a,a(x))/n$.
\end{proof}

\subsection{A reusable base bound}

For later use, define a common upper bound for the three base-PMW terms.  For
arbitrary
$J\ge1$, $\eta>0$, and
$0<\gamma\le1$, define
\[
  \alpha_0
  \defeq
  \frac{\log T}{\gamma J}+\frac\gamma2
  +\frac{4\log(4k)}{\eta}.
  \tag{3.12}\label{eq:alpha-zero}
\]
The constants leave enough slack for both the $L_2$ estimate here and the
integer likelihood-ratio moments used later.  In particular,
\cref{prop:robust-utility} and $L_1\le L_2$ give
\[
  \norm{\ell_x}_{L_2(P_{a(x)})}\le\alpha_0,
  \qquad
  \E_{P_a}\ell_x
  \le\alpha_0+\frac{4d_*(a,a(x))}{n}.
  \tag{3.13}\label{eq:base-utility}
\]
These inequalities hold for every admissible target $a$; the target need
not be realized by a dataset.  The explicit choices of $J,\eta,$ and
$\gamma$ are postponed to \cref{sec:parameters}.

\section{The transcript envelope}
\label{sec:envelope}

We now define the final private transcript mechanism.  For each possible
transcript, take its likelihood under every dataset and discount that
likelihood by the dataset's Hamming distance from $x$.  The largest discounted
likelihood is the envelope at $x$.

For dataset targets, abbreviate $P_x=P_{a(x)}$ and $p_x=p_{a(x)}$.  Let
$\lambda=\eps/2$ and define
\[
  \widetilde p_x(\omega)
  =\max_{y\in\cD^n}e^{-\lambda d_H(x,y)}p_y(\omega),
  \qquad
  Z_x=\sum_{\omega\in\Om}\widetilde p_x(\omega).
  \tag{4.1}\label{eq:envelope}
\]
The final transcript mass function, transcript law, and released answer are
\[
  \widehat p_x(\omega)=\frac{\widetilde p_x(\omega)}{Z_x},
  \qquad \omega\sim\widehat P_x,
  \qquad M(x)=\Ans(\omega),
  \tag{4.2}\label{eq:final-mechanism}
\]
where $\widehat P_x$ denotes the distribution with mass function
$\widehat p_x$.
Both maxima and sums are over finite nonempty sets.  Moreover $Z_x\ge1$
because the maximum includes $y=x$ and $p_x$ sums to one.
Every base transcript has positive mass under every target because each
selection is a softmax draw.  The construction therefore keeps common support
throughout; it never conditions on a data-dependent ``good'' event.

\begin{proposition}[Pure privacy of the envelope]
\label{prop:envelope-privacy}
The transcript law $\widehat P$ in \eqref{eq:final-mechanism} is pure
$\eps$-differentially private.  Hence so is $M$.
\end{proposition}

\begin{proof}
If $d_H(x,x')=1$, the triangle inequality gives, pointwise in $\omega$,
\[
 e^{-\lambda}\widetilde p_{x'}(\omega)
 \le\widetilde p_x(\omega)
 \le e^\lambda\widetilde p_{x'}(\omega).
\]
Summing gives the same two-sided comparison between $Z_x$ and $Z_{x'}$.
Consequently
$e^{-2\lambda}\le\widehat p_x(\omega)/\widehat p_{x'}(\omega)
\le e^{2\lambda}$.  Since $2\lambda=\eps$, summing this pointwise bound over
any event proves pure DP.  The answer is deterministic postprocessing of the
transcript.
\end{proof}

Privacy did not use an upper bound on $Z_x$.  The rest of the paper proves
the stronger pair of estimates
\[
  Z_x=O(1),\qquad
  \sum_\omega\widetilde p_x(\omega)\ell_x(\omega)
  =O\!\left(\alpha_0+\frac1{\eps n}\right).
  \tag{4.3}\label{eq:envelope-goals}
\]

\section{Likelihood-level domination for a Hamming ball}
\label{sec:maurey}

The envelope in \eqref{eq:envelope} ranges over all $T^n$ datasets.  For its
utility analysis, we pointwise upper-bound the likelihood maximum over every
database in a Hamming ball by a small family of answer-space surrogates.  For
each competing
database $y$, a $y$-dependent random surrogate supported on one family
depending only on $x$, $R$, and a code length $m$ has mean $a(y)$.  Thus the
same support family works for every $y$ in the ball.  Its size is at most
$(T^2+1)^m$, independent of $n$.

The sampling device is a finite $\ell_\infty$ instance of Maurey's empirical
method \cite{pisier1981maurey}.  Here it is applied directly to likelihoods:
a fixed-transcript log-sum-exp and Hoeffding estimate implies that some
surrogate, possibly depending on the transcript, has likelihood nearly as
large as the database it represents.

Fix a dataset $x$, an integer radius $R\ge1$, and an integer code length
$m\ge1$.  For $c=(d,d')\in\cD\times\cD$, define the admissible row-move vector
\[
  \delta_c(h)=h(d')-h(d)\qquad(h\in\cS),
\]
and let $\delta_0=0$.  Define the finite target family
\[
 \cV_{x,R,m}
 =\left\{a(x)+\frac{R}{mn}\sum_{i=1}^m\delta_{c_i}:
 c_i\in\{0\}\mathbin{\dot\cup}(\cD\times\cD)\right\}.
 \tag{5.1}\label{eq:maurey-family}
\]
Every $v\in\cV_{x,R,m}$ satisfies $d_*(v,a(x))\le R$, and
\[
  \abs{\cV_{x,R,m}}\le (T^2+1)^m.
  \tag{5.2}\label{eq:maurey-card}
\]
The disjoint union notation treats the zero move as a separate label;
different labels may define the same target, which only improves the bound.

\begin{lemma}[Likelihood Maurey bound]
\label{lem:likelihood-maurey}
For every $y\in\cD^n$ with $d_H(x,y)\le R$ and every transcript
$\omega\in\Om$,
\[
  p_y(\omega)
  \le \exp\!\left(\frac{\kappa R^2}{m}\right)
  \max_{v\in\cV_{x,R,m}}p_v(\omega).
  \tag{5.3}\label{eq:likelihood-maurey}
\]
\end{lemma}

\begin{proof}
Let $h=d_H(x,y)$.  Make a list of length $R$ containing the $h$ row-move
vectors
\[
  \bigl(s(y_i)-s(x_i)\bigr)_{s\in\cS}
  \qquad(x_i\ne y_i)
\]
and $R-h$ zero vectors.  Draw $W_1,\ldots,W_m$ independently and uniformly
from this list, and put
\[
  V=a(x)+\frac{R}{mn}\sum_{i=1}^mW_i.
  \tag{5.4}\label{eq:random-target}
\]
Then $\E V=a(y)$ coordinatewise and $V$ is supported on
$\cV_{x,R,m}$.

Fix $\omega=(s_1,\ldots,s_J)$.  Once the full transcript is fixed, every
synthetic iterate $\mu_{t-1}$ is fixed.  The negative log likelihood
$g_\omega(a)=-\log p_a(\omega)$ has the form
\[
 g_\omega(a)=C_\omega+
 \sum_{t=1}^J\left[
 \log\sum_{u\in\cS}e^{\eta(a_u-u(\mu_{t-1}))}
 -\eta a_{s_t}\right],
 \tag{5.5}\label{eq:fixed-loglik}
\]
where $C_\omega=\eta\sum_ts_t(\mu_{t-1})$ is independent of $a$.

Consider one log-partition and let $w_u$ be its softmax weights at $a(y)$.
Jensen's inequality and then Hoeffding's lemma give
\[
\begin{aligned}
 &\E\log\sum_{u\in\cS}w_u
       e^{\eta(V_u-a(y)_u)}\\
 &\quad\le
 \log\sum_{u\in\cS}w_u\E e^{\eta(V_u-a(y)_u)}
 \le\frac{2\eta^2R^2}{mn^2}.
\end{aligned}
\tag{5.6}\label{eq:partition-maurey}
\]
Indeed, for fixed $u$,
\[
 V_u-a(y)_u
 =\frac{R}{mn}\sum_{i=1}^m(W_i(u)-\E W_i(u)),
\]
and each $W_i(u)$ lies in $[-2,2]$, an interval of length four.  Because
$\E V=a(y)$, the expectation of each selected term $-\eta V_{s_t}$ is exactly
$-\eta a(y)_{s_t}$.  Sum \eqref{eq:partition-maurey} over $J$ rounds and use
the identity $\kappa=2J\eta^2/n^2$ to obtain
\[
  \E g_\omega(V)
  \le g_\omega(a(y))+\frac{\kappa R^2}{m}.
\]
Equivalently,
\[
 \log p_y(\omega)
 \le\E\log p_V(\omega)+\frac{\kappa R^2}{m}.
\]
Exponentiation and the geometric-mean/arithmetic-mean inequality give
\[
 p_y(\omega)
 \le e^{\kappa R^2/m}\E p_V(\omega)
 \le e^{\kappa R^2/m}
      \max_{v\in\cV_{x,R,m}}p_v(\omega).
\]
This is \eqref{eq:likelihood-maurey}.
\end{proof}

\begin{remark}
Because \eqref{eq:likelihood-maurey} is pointwise in $\omega$, it may be
multiplied by any nonnegative weight and then summed.  Taking that weight to
be $\ell_x$ is exactly the error-weighted use made in \cref{sec:maximal}; the
lemma controls the utility quantity needed for the envelope, not only its
added mass.
The Maurey bound is independent of $k$ because it uses no union bound over
queries.  Fixing the complete selected-query transcript before viewing
\eqref{eq:fixed-loglik} as a function of the target absorbs adaptivity without
an additional term.
\end{remark}

\section{Maximal likelihood on a Hamming ball}
\label{sec:maximal}

Maurey domination leaves a maximum over a finite family of target vectors.
We must bound both its total mass and its error-weighted mass.  For nearby
targets, likelihood-ratio moments control the maximum.  For more distant
targets, we instead upper-bound the maximum by a sum of probability laws.

Fix a dataset $x$ and an integer radius $R\in\mathbb N$ with $R\ge1$.  Write
\[
  H_R(\omega)
  =\max_{y:d_H(x,y)\le R}p_y(\omega),
  \qquad
  L=\log(T^2+1),
\]
and introduce the transcript moment coefficient $\kappa$ and the resulting
radius-$R$ moment scale $A_R$:
\[
  \kappa=\frac{2J\eta^2}{n^2},
  \qquad
  A_R=\kappa R^2.
  \tag{6.1}\label{eq:ball-scales}
\]
Thus the Maurey factor in \eqref{eq:likelihood-maurey} is
$e^{A_R/m}$.  The quantity $H_R$ is the maximum over the closed ball and will
upper-bound each annulus whose outer radius is $R$.

\begin{lemma}[Integer likelihood-ratio moment]
\label{lem:integer-likelihood-moment}
If $v$ is an admissible target with $d_*(v,a(x))\le R$, then every integer
$r\ge1$ satisfies
\[
 \sum_{\omega}p_x(\omega)
 \left(\frac{p_v(\omega)}{p_x(\omega)}\right)^r
 \le e^{A_Rr^2}.
 \tag{6.2}\label{eq:natural-moment}
\]
\end{lemma}

\begin{proof}
Condition on a transcript prefix at round $t$, and write $\pi_x^t$ and
$\pi_v^t$ for the resulting softmax kernels.  Since the synthetic iterate is
determined by the prefix, their logits differ only through
$\delta_s=v_s-a(x)_s\in[-2R/n,2R/n]$.  Therefore

\[
  \frac{\pi_v^t(s)}{\pi_x^t(s)}
  =\frac{e^{\eta\delta_s}}{\E_{\pi_x^t}e^{\eta\delta}}.
\]
Jensen's inequality lower-bounds the denominator by
$e^{\eta\E_{\pi_x^t}\delta}$.  Hoeffding's lemma therefore gives

\[
  \E_{\pi_x^t}\left(\frac{\pi_v^t}{\pi_x^t}\right)^r
  \le \E_{\pi_x^t}e^{r\eta(\delta-\E_{\pi_x^t}\delta)}
  \le \exp\left(\frac{2r^2\eta^2R^2}{n^2}\right).
\]
Iterating conditional expectation over $J$ rounds multiplies these moment
costs, giving \eqref{eq:natural-moment} because
$A_R=2J\eta^2R^2/n^2$.
\end{proof}

\subsection{Nearby databases: likelihood-ratio moments}
\label{sec:near-max}

Suppose $A_R\le L$.  Apply \cref{lem:likelihood-maurey} with code length
$m=1$, and write $\cV=\cV_{x,R,1}$.  Set
\[
  r_R=\left\lceil\sqrt{\frac{L}{A_R}}\right\rceil+1.
  \tag{6.3}\label{eq:near-orders}
\]
Then $r_R\ge2$,
\[
 \sqrt{\frac{L}{A_R}}\le r_R
 <\sqrt{\frac{L}{A_R}}+2,
\]
and consequently
\[
  \frac{L}{r_R}+A_Rr_R
  \le4\sqrt{A_RL}.
  \tag{6.4}\label{eq:near-optimization}
\]
Indeed, the first term is at most $\sqrt{A_RL}$, while
$A_Rr_R\le\sqrt{A_RL}+2A_R\le3\sqrt{A_RL}$.

Put $\Lambda_v=p_v/p_x$; the common softmax support makes this ratio well
defined.
The Maurey bound and \eqref{eq:natural-moment} give
\[
 H_R\le e^{A_R}p_x\max_{v\in\cV}\Lambda_v,
 \qquad
 \left\lVert\max_{v\in\cV}\Lambda_v\right\rVert_{L_{r_R}(P_x)}
 \le\left(\sum_{v\in\cV}\E_{P_x}\Lambda_v^{r_R}\right)^{1/r_R}
 \le e^{L/r_R+A_Rr_R}.
\]
For the conjugate exponent $b_R=r_R/(r_R-1)\le2$, Hölder's inequality and
monotonicity of probability-space norms give
$\lVert \phi\rVert_{L_{b_R}(P_x)}\le\lVert \phi\rVert_{L_2(P_x)}$.
Therefore every nonnegative $\phi$ satisfies
\[
 \sum_\omega H_R(\omega)\phi(\omega)
 \le
 \exp\left\{
   A_R+\frac{L}{r_R}+A_Rr_R
 \right\}
 \norm{\phi}_{L_2(P_x)}.
 \tag{6.5}\label{eq:near-max}
\]
The three exponents are respectively the Maurey, cardinality, and
likelihood-moment costs.  Since $A_R\le L$,
$A_R+4\sqrt{A_RL}\le5\sqrt{A_RL}$.  Taking $\phi=1$ and
$\phi=\ell_x$,
respectively, gives
\[
 \sum_\omega H_R(\omega)
 \le e^{5\sqrt{A_RL}},
 \qquad
 \sum_\omega H_R(\omega)\ell_x(\omega)
 \le e^{5\sqrt{A_RL}}\alpha_0.
 \tag{6.6}\label{eq:near-loss}
\]

\subsection{Distant databases: a mixture bound}
\label{sec:far-max}

Suppose instead that $A_R>L$, and choose
\[
 m_R=\max\left\{1,
   \left\lceil\sqrt{\frac{A_R}{L}}\right\rceil
 \right\}.
 \tag{6.7}\label{eq:far-m}
\]
After the Maurey step, use $\max\le\sum$.  Since the code family has at most
$e^{m_RL}$ members,
\[
 \sum_\omega H_R(\omega)
 \le e^{A_R/m_R+m_RL}.
 \tag{6.8}\label{eq:far-mass}
\]
Every code target lies within target distance $R$ of the center, so
\eqref{eq:base-utility} gives
$\E_{P_v}\ell_x\le\alpha_0+4R/n$.  Hence
\[
 \sum_\omega H_R(\omega)\ell_x(\omega)
 \le
 e^{A_R/m_R+m_RL}
 \left(\alpha_0+\frac{4R}{n}\right).
 \tag{6.9}\label{eq:far-loss}
\]
The ceiling bounds imply
\[
 \frac{A_R}{m_R}+m_RL
 \le2\sqrt{A_RL}+L
 \le3\sqrt{A_RL},
 \tag{6.10}\label{eq:far-optimized}
\]
where the last inequality uses $L\le A_R$.

Combining the near and far cases gives the uniform ball estimate used for
every radial block:
\[
 \boxed{
 \begin{aligned}
  \sum_\omega H_R(\omega)
  &\le e^{5\sqrt{A_RL}},\\
  \sum_\omega H_R(\omega)\ell_x(\omega)
  &\le
  e^{5\sqrt{A_RL}}
  \left(\alpha_0+\frac{4R}{n}\right).
 \end{aligned}}
 \tag{6.11}\label{eq:one-ball-bound}
\]

\section{Radial blocking and envelope utility}
\label{sec:blocking}

We now sum \eqref{eq:one-ball-bound} in blocks of radius
$\Theta(1/\eps)$, the scale over which the envelope penalty changes by a
constant factor.  Fix
\[
  c_0=\frac1{900},
  \qquad
  \kappa=\frac{2J\eta^2}{n^2}
  \le\frac{c_0\eps^2}{L}.
  \tag{7.1}\label{eq:kappa-choice}
\]
This condition makes the growth of every ball likelihood small compared with
the envelope's distance penalty.

\paragraph{Log-cost decomposition.}
The following ledger summarizes the preceding argument for the radius-$R$
ball.  Its entries are additive contributions to the logarithm of the ball
bound; $m_R$ and $r_R$ are the optimized code length and moment order from
\cref{sec:maximal}.

\begin{table}[H]
  \centering
  \caption{Log-cost ledger for the radius-$R$ ball bound.}
  \label{tab:radius-ledger}
  \begin{tabular}{@{}
    >{\raggedright\arraybackslash}p{0.30\linewidth}
    >{\centering\arraybackslash}p{0.27\linewidth}
    >{\centering\arraybackslash}p{0.27\linewidth}@{}}
    \toprule
    Contribution & Near: moment & Far: mixture \\
    \midrule
    Maurey smoothing  & $A_R$ & $A_R/m_R$ \\
    Code cardinality  & $L/r_R$ & $m_RL$ \\
    Likelihood moment & $A_Rr_R$ & none \\
    Envelope penalty  & $-\eps R/2$ & $-\eps R/2$ \\
    \bottomrule
  \end{tabular}
\end{table}

When $A_R\le L$, one Maurey move and
$r_R\simeq\sqrt{L/A_R}$ make the positive cost at most
$5\sqrt{A_RL}$; the error factor is $\alpha_0$.  When $A_R>L$,
$m_R\simeq\sqrt{A_R/L}$ makes the positive cost at most
$3\sqrt{A_RL}$; after $\max\le\sum$, each mixture law has total mass one,
so there is no likelihood-moment cost, and the prefactor is
$\alpha_0+4R/n$.  By \eqref{eq:kappa-choice}, both positive costs are at most
$\eps R/6$, below the envelope penalty.

A radius-by-radius summation would incur an additional factor $1/\eps$.
Indeed, even an exact-radius estimate
decaying as $e^{-\eps R/3}$ would give, for $0<\eps\le1$,
\[
  \sum_{R\ge1}e^{-\eps R/3}=\Theta(1/\eps),
  \qquad
  \sum_{R\ge1}\frac{R}{n}e^{-\eps R/3}
  =\Theta\!\left(\frac1{\eps^2n}\right).
\]
This loss comes from discretizing radii, not from a union bound over datasets,
which Maurey domination has already avoided.
Blocking at width $B\simeq1/\eps$ removes exactly this discretization loss:
the rigorous argument uses the outer radius of a block for the positive ball
cost and its inner radius for the negative envelope penalty.  The resulting
block weights have constant total mass, and their $(R/n)$-weighted
contribution is $O(B/n)=O(1/(\eps n))$.

Set the integer block width
\[
  B=\max\{1,\lceil3/\eps\rceil\},
  \qquad \theta=\eps B.
 \tag{7.2}\label{eq:block-width}
\]
Then $\theta\ge3$.  If $0<\eps\le3$, then $\theta\le6$; if
$\eps>3$, then $B=1$ and $\theta=\eps$.

For $j\ge0$, let
\[
  R_j=(j+1)B.
 \tag{7.3}\label{eq:annuli}
\]
The envelope-growth condition gives
\[
  5\sqrt{A_{R_j}L}
  \le5\sqrt{c_0}\,\eps R_j
  =5\sqrt{c_0}\,\theta(j+1)
  \le\frac{\theta(j+1)}6,
  \tag{7.4}\label{eq:annulus-A}
\]
because $\sqrt{c_0}=1/30$.
Thus \eqref{eq:one-ball-bound} becomes
\[
 \begin{aligned}
  \sum_\omega H_{R_j}(\omega)
  &\le
    \exp\left\{\frac{\theta(j+1)}6\right\},\\
  \sum_\omega H_{R_j}(\omega)\ell_x(\omega)
  &\le
    \exp\left\{\frac{\theta(j+1)}6\right\}
    \left(\alpha_0+\frac{4B(j+1)}{n}\right).
 \end{aligned}
 \tag{7.5}\label{eq:shell-ball-bound}
\]

Every $y\ne x$ belongs to a unique block indexed by
\[
  j=\left\lfloor\frac{d_H(x,y)-1}{B}\right\rfloor<n
\]
and satisfies $jB+1\le d_H(x,y)\le(j+1)B$.  Consequently, pointwise in the
transcript,
\[
  \widetilde p_x(\omega)
  \le p_x(\omega)+
  \sum_{j=0}^{n-1}
    e^{-(\eps/2)(jB+1)}H_{R_j}(\omega).
  \tag{7.6}\label{eq:envelope-block-decomposition}
\]

Define the large-$\eps$ damping factor
\[
  u_\eps=
  \begin{cases}
    1,&0<\eps\le3,\\
    e^{-\eps/3},&\eps>3,
  \end{cases}
  \tag{7.7}\label{eq:prefix}
\]
and define the combined block weight
\[
  \Phi_j
  =
  \exp\left\{-\frac{\eps}{2}(jB+1)\right\}
  \exp\left\{\frac{\theta(j+1)}6\right\}.
  \tag{7.8}\label{eq:complete-shell-factor}
\]

\begin{lemma}[Finite blocked-radius summation]
\label{lem:blocked-shell-sum}
For every $j\ge0$,
\[
  \Phi_j
  \le eu_\eps e^{-j\theta/3}.
  \tag{7.9}\label{eq:pointwise-shell-sum}
\]
Consequently, for every integer $N\ge0$,
\[
  \sum_{j=0}^{N-1}\Phi_j
  \le2eu_\eps,
  \qquad
  \sum_{j=0}^{N-1}(j+1)\Phi_j
  \le4eu_\eps.
  \tag{7.10}\label{eq:geometric-sums}
\]
\end{lemma}

\begin{proof}
See \cref{app:proof-blocked-shell-sum}.
\end{proof}

The damping factor also satisfies
\[
  u_\eps\le1,
  \qquad
  Bu_\eps\le\frac9\eps.
  \tag{7.11}\label{eq:block-eps}
\]
For $\eps\le3$, this uses $B\le6/\eps$.  For $\eps>3$, it uses $B=1$
and the elementary bound $\eps e^{-\eps/3}\le9$.

We can now sum the block bound.  The central base law has mass
one and expected loss at most $\alpha_0$.  Combining
\eqref{eq:envelope-block-decomposition}, \eqref{eq:shell-ball-bound}, and
\cref{lem:blocked-shell-sum} gives
\[
\begin{aligned}
  \sum_\omega\widetilde p_x(\omega)\ell_x(\omega)
  &\le
  \alpha_0+
  \alpha_0\sum_{j=0}^{n-1}\Phi_j+
  \frac{4B}{n}\sum_{j=0}^{n-1}(j+1)\Phi_j\\
  &\le
  \alpha_0+2eu_\eps\alpha_0+
  \frac{16eBu_\eps}{n}\\
  &\le
  (1+2e)\alpha_0+\frac{144e}{\eps n}.
\end{aligned}
\tag{7.12}\label{eq:unnormalized-error}
\]
Here $2e$ and $4e$ come from \eqref{eq:geometric-sums}, while
$144e=16e\cdot9$ also uses \eqref{eq:block-eps}.

Repeating the same calculation with $\ell_x$ replaced by $1$ gives the
normalizer estimate
\[
  1\le Z_x\le1+2eu_\eps\le1+2e.
  \tag{7.13}\label{eq:normalizer}
\]
The lower bound follows from $\widetilde p_x\ge p_x$.  Since $Z_x\ge1$,
normalization can only decrease the unnormalized error bound, and therefore
\[
  \E\norm{M(x)-F(x)}_\infty
  \le(1+2e)\alpha_0+\frac{144e}{\eps n}.
  \tag{7.14}\label{eq:pre-final-rate}
\]

\section{Parameter regimes and proof of the main theorem}
\label{sec:parameters}

Ignoring constants, the parameter balance is as follows.  To target
error $\alpha$, take
$\gamma=\Theta(\alpha)$,
$J=\Theta(\log T/\alpha^2)$, and
$\eta=\Theta(\log k/\alpha)$.  The first two choices balance Hedge regret,
and the third controls the exponential-mechanism selection error.  Then
$\alpha_0=O(\alpha)$, while the envelope-growth condition reduces to
$\alpha^2=\Omega(\log T\log k/(\eps n))$.

We now verify explicit integer choices.  The three numerical parameters below
were selected jointly against the resulting center and composition terms and
then rounded to simple rational values; no best-possible universal constant is
claimed.  Put
\begin{align}
  L_D&=\log(2T),\qquad L_Q=\log(2k),
  \tag{8.1}\label{eq:positive-logs}\\
  34200L_DL_Q&\le \eps n\alpha^2,
  \qquad 0<\alpha\le1.
  \tag{8.2}\label{eq:explicit-sample-condition}
\end{align}
The positive logarithms exceed $1/2$, including when $T=1$ or $k=1$.
Choose
\begin{align}
  J&=\left\lceil\frac{98L_D}{\alpha^2}\right\rceil,
  &\eta&=\frac{57L_Q}{\alpha},
  &\gamma&=\frac\alpha7,
  \tag{8.3}\label{eq:lean-parameters}\\
  \frac{98L_D}{\alpha^2}&\le J
  \le\frac{100L_D}{\alpha^2}.
  \tag{8.4}\label{eq:round-bounds}
\end{align}
For the base PMW error $\alpha_0$ defined in \eqref{eq:alpha-zero}, the three
displayed contributions are at most $\alpha/14,\alpha/14,$ and
$8\alpha/57$, respectively.  Thus
\[
  \alpha_0\le\frac{113}{399}\alpha.
\]

For $L=\log(T^2+1)$ and $\kappa=2J\eta^2/n^2$, the round bound,
$L\le2L_D$, and the square of \eqref{eq:explicit-sample-condition} give
\[
  \kappa\le\frac{649800L_DL_Q^2}{n^2\alpha^4},
  \qquad
  900\kappa L
  \le\frac{34200^2L_D^2L_Q^2}{n^2\alpha^4}
  \le\eps^2.
\]
Hence the envelope-growth condition holds:
\[
  \frac{2J\eta^2}{n^2}
  \le \frac{\eps^2}{900\log(T^2+1)}.
  \tag{8.5}\label{eq:composition-check}
\]
Combining this with \eqref{eq:pre-final-rate} gives
\begin{align}
  \alpha_0&\le\frac{113}{399}\alpha,
  \tag{8.6}\label{eq:linear-absorbed}\\
  \E\norm{M(x)-F(x)}_\infty
  &\le(1+2e)\alpha_0+\frac{144e}{\eps n}.
  \tag{8.7}\label{eq:explicit-alpha-rate}
\end{align}

Finally set
$\tau=\sqrt{L_DL_Q/(\eps n)}$.
If $185\tau\le1$, take $\alpha=185\tau$.  Then
$\eps n\alpha^2=34225L_DL_Q$, so
\eqref{eq:explicit-sample-condition} holds.  Moreover,
$L_DL_Q\ge1/4$ gives
\[
  \frac1{\eps n}
  =\frac{\tau^2}{L_DL_Q}
  \le4\tau^2\le\frac{4\tau}{185}.
\]
Consequently \eqref{eq:explicit-alpha-rate} is at most
\[
  (1+2e)\frac{113}{399}(185\tau)
  +\frac{576e}{185}\tau
  \le129e\tau,
\]
where the last inequality uses $e>2.7$.
If $185\tau>1$, use the separate deterministic zero-output mechanism.  Its
error is at most one and its one-point output law is $0$-DP; also
$1\le129e\min\{1,\tau\}$.  Thus \cref{thm:main} holds in all regimes
with the stated constant $129e$.  The nontrivial branch is private by
\cref{prop:envelope-privacy}, followed by deterministic postprocessing.

\section{An envelope proof of the normalized Euclidean rate}
\label{sec:l2}

We now prove the secondary \cref{thm:main-l2}.  Nikolov's efficient
Johnson--Lindenstrauss mechanism already resolved Open Problem~2
\cite{nikolov2023jl}; the purpose here is to show that the same rate also
falls within the envelope framework and admits an end-to-end formalization.
The envelope, likelihood-level Maurey
argument, and radial blocking are unchanged.  The only modification is the
base transcript: instead of privately selecting the query and its sign, each
round samples the query uniformly, independently of the data, and privately
selects only its sign.  This replaces the selection cost $\log k$ by an
average squared-residual estimate.

Write
\[
  \norm{z}_{2,k}=\frac1{\sqrt{k}}\norm{z}_2.
\]
For an auxiliary target vector $a\in\R^{\cQ}$ and current synthetic
distribution $\mu$, sample $q$ uniformly from $\cQ$ and then sample
$\sigma\in\{-1,1\}$ according to
\[
  \Prb_a(\sigma\mid q,\mu)
  =\frac{\exp\{(\sigma/4)(a_q-q(\mu))\}}
  {\exp\{(a_q-q(\mu))/4\}+\exp\{-(a_q-q(\mu))/4\}}.
  \tag{9.1}\label{eq:l2-sign-kernel}
\]
The update remains \eqref{eq:hedge-update}, with $s=\sigma q$, and the
decoder again outputs the query answers of the average synthetic
distribution.  Denote the resulting transcript law by $P_a^{(2)}$, its mass
function by $p_a^{(2)}$, and its normalized Euclidean decoding loss relative
to $x$ by
\[
  \ell_x^{(2)}(\omega)
  =\norm{\widehat a(\omega)-F(x)}_{2,k}.
\]

The conditional mean of the selected sign is
\[
  \E_a[\sigma\mid q,\mu]
  =\tanh\!\left(\frac{a_q-q(\mu)}4\right).
  \tag{9.2}\label{eq:l2-sign-mean}
\]
The following exact scalar inequality supplies the required progress.  If
$\abs r\le2$, then, for every $\delta\in\R$,
\[
  \frac{r^2}{24}-\frac38\delta^2
  \le r\tanh\!\left(\frac{r+\delta}{4}\right).
  \tag{9.3}\label{eq:robust-tanh}
\]
Indeed, $r^2/12\le r\tanh(r/4)$, the derivative bound
$\abs{\tanh'(u)}\le1$ is global, and
$\abs{r\delta}/4\le r^2/24+3\delta^2/8$ by Young's inequality.
Thus \eqref{eq:robust-tanh} uses no Taylor approximation or asymptotic
relaxation.

Suppose $\max_q\abs{a_q-F_q(x)}\le B$.  Apply
\eqref{eq:robust-tanh} with
$r=F_q(x)-q(\mu)$ and $\delta=a_q-F_q(x)$, average over the uniform query,
and sum the Hedge potential bound over the adaptive transcript.  Convexity of
squared loss then passes from the trajectory to the averaged decoder and
gives
\[
  \E_{P_a^{(2)}}\bigl[(\ell_x^{(2)})^2\bigr]
  \le
  24\left(\frac{\log T}{\gamma J}+\frac\gamma2\right)+9B^2.
  \tag{9.4}\label{eq:l2-robust-utility}
\]
In particular, at the center $a=F(x)$, both the expected loss and its
$L_2(P_{F(x)}^{(2)})$ norm are at most
\[
  \alpha_2(\gamma,J)
  :=\sqrt{24\left(\frac{\log T}{\gamma J}+\frac\gamma2\right)}.
  \tag{9.5}\label{eq:l2-center}
\]

The likelihood calculation is simpler than for signed-query selection.
The uniform query factor cancels exactly between two target vectors.  The
conditional sign law in \eqref{eq:l2-sign-kernel} is a two-point exponential
tilt, so Hoeffding's lemma and adaptive composition give, for every positive
integer $r$ and every $\max_q\abs{a_q-b_q}\le B$,
\[
  \sum_\omega p_b^{(2)}(\omega)
    \left(\frac{p_a^{(2)}(\omega)}{p_b^{(2)}(\omega)}\right)^r
  \le
  \exp\!\left\{\frac{JB^2}{2}\left(\frac r4\right)^2\right\}.
  \tag{9.6}\label{eq:l2-integer-moment}
\]
Thus the likelihood-level Maurey proof of
\cref{lem:likelihood-maurey} applies verbatim with
\[
  \kappa_2=\frac{J}{8n^2}.
  \tag{9.7}\label{eq:l2-kappa}
\]
In particular, sampling $m$ row replacements dominates the likelihood of a
radius-$R$ target at pointwise cost $\exp\{\kappa_2R^2/m\}$.

For completeness, let
\[
  H_R^{(2)}(\omega)
  =\max_{y:d_H(x,y)\le R}p_{F(y)}^{(2)}(\omega),
  \qquad L=\log(T^2+1),
\]
and put $\overline R=\min\{R,n\}$.  Optimizing the same nearby moment bound
and distant mixture bound as in \cref{sec:near-max,sec:far-max} yields
\begin{align*}
  \sum_\omega H_R^{(2)}(\omega)
  &\le \exp\{5\overline R\sqrt{\kappa_2L}\},
  \tag{9.8}\label{eq:l2-ball-mass}\\
  \sum_\omega H_R^{(2)}(\omega)\ell_x^{(2)}(\omega)
  &\le \exp\{5\overline R\sqrt{\kappa_2L}\}
    \left(\alpha_2(\gamma,J)+
      \frac{6\overline R}{n}\right).
  \tag{9.9}\label{eq:l2-ball-loss}
\end{align*}
The linear radial term follows from \eqref{eq:l2-robust-utility} with
$B=2R/n$ and $\sqrt{u+9B^2}\le\sqrt u+3B$.

Use $p_{F(y)}^{(2)}$ in the same envelope \eqref{eq:intro-envelope} and
normalize.  If
\[
  \kappa_2\le\frac{\eps^2}{900L},
  \tag{9.10}\label{eq:l2-composition}
\]
then the blocked-radius proof of \cref{sec:blocking}, with all radii capped at
$n$, combines \eqref{eq:l2-ball-mass} and \eqref{eq:l2-ball-loss}.  The
geometric and weighted-geometric sums are unchanged.  Retaining the center
term separately gives the explicit normalized-envelope bound
\[
  \E\ell_x^{(2)}
  \le(1+2e)\alpha_2(\gamma,J)
    +\frac{216e}{\eps n}.
  \tag{9.11}\label{eq:l2-envelope-rate}
\]
The envelope privacy proof is exactly \cref{prop:envelope-privacy}, and the
decoded vector mechanism is pure $\eps$-DP by postprocessing.

It remains to choose parameters.  Put $L_D=\log(2T)$ and, for
$0<\alpha\le1$, assume
\[
  2162L_D\le\eps n\alpha^2.
  \tag{9.12}\label{eq:l2-sample-condition}
\]
Choose
\[
  J=\left\lceil\frac{20000L_D}{\alpha^4}\right\rceil,
  \qquad \gamma=\frac{\alpha^2}{96}.
  \tag{9.13}\label{eq:l2-parameters}
\]
Then $J\le20002L_D/\alpha^4$,
$\alpha_2(\gamma,J)\le\alpha/2$, and
\eqref{eq:l2-sample-condition} implies \eqref{eq:l2-composition}.
Substitution into \eqref{eq:l2-envelope-rate} gives
\[
  \E\ell_x^{(2)}
  \le\frac{1+2e}{2}\alpha+\frac{216e}{\eps n}.
  \tag{9.14}\label{eq:l2-alpha-rate}
\]

Finally let
$\tau_2=\sqrt{L_D/(\eps n)}$.  If $(93/2)\tau_2\le1$, take
$\alpha=(93/2)\tau_2$.  Since $(93/2)^2=8649/4>2162$, the sample condition
holds.  Moreover, $L_D\ge\log2>0.69$ and $\tau_2\le2/93$ imply
$1/(\eps n)=\tau_2^2/L_D\le\tau_2/32$.  Hence
\eqref{eq:l2-alpha-rate} is at most
\[
  \frac{93}{4}\tau_2+
  \left(\frac{93}{2}+\frac{27}{4}\right)e\tau_2
  \le62e\tau_2,
\]
where the last inequality uses $e>2.7$.  If $(93/2)\tau_2>1$, use the
deterministic zero-output mechanism.  Its normalized
Euclidean error is at most one and it is $0$-DP, while
$1\le62e\min\{1,\tau_2\}$.  This proves \cref{thm:main-l2} in all regimes.

\section{Discussion and limitations}

The main theorem addresses the expected-$\ell_\infty$ formulation of Open
Problem~1; \cref{thm:main-l2} independently recovers the previously resolved
Open Problem~2 rate.  Markov's inequality gives constant-probability
accuracy at the same respective asymptotic rates.  For worst-coordinate
failure probability $\zeta$,
running $O(\log(1/\zeta))$ independent copies with the privacy budget divided
among them and taking a coordinatewise median preserves pure $\eps$-DP and
incurs an additional
$O(\bigl(\log(1/\zeta)\bigr)^{1/2})$ error factor.  The single-run envelope
analysis does not establish optimal confidence dependence.

The mechanism is information-theoretic.  The envelope maximizes over
$\cD^n$, and its transcript space has size $(2k)^J$; the proof does not yield
a polynomial-time implementation.  Finding an efficient pure-DP mechanism
with the same universal rate remains open.

Finally, \cref{thm:main,thm:main-l2} are uniform over all finite parameter
choices, but they are not all-regimes minimax statements.  The outer minima and shifted
logarithms cover degenerate cases correctly; they do not assert matching
lower bounds there.  For instance, a singleton universe permits exact
release, and a single query admits the sensitivity-calibrated Laplace
mechanism.  More generally, direct vector release with independent Laplace
noise has expected worst-coordinate error
$O(k\log(2k)/(\eps n))$ and normalized Euclidean error
$O(k/(\eps n))$, which are preferable in some small-workload regimes.

\section{Formal verification}
\label{sec:formal}

Sections~3--9 together with Appendix~A give standalone analytic proofs.  A
companion Lean 4 development, pinned to Lean and Mathlib 4.28.0, formalizes
\cref{thm:main,thm:main-l2}, including both finite mechanisms and the
integer-order likelihood calculations used in the proofs.  It verifies
normalization, pure privacy before and after each dataset-independent decoder,
the integer parameter choices, and both explicit all-regimes expected-error
bounds.  The declaration
\nolinkurl{optimal_pureDP_decoded_statistical_query_release_optimized_rate}
has constant
\[
  129e\min\left\{1,
  \sqrt{\frac{\log(2T)\log(2k)}{\eps n}}\right\}.
\]
The normalized-Euclidean bound is formalized by the Lean declaration
\begin{center}
  \small
  \nolinkurl{optimal_pureDP_decoded_normalized_l2_statistical_query_release_quadratic_rate}.
\end{center}
Its constant is
\[
  62e\min\left\{1,
  \sqrt{\frac{\log(2T)}{\eps n}}\right\}.
\]
These declarations are \cref{thm:main,thm:main-l2}: they quantify over arbitrary finite
nonempty $\cD$ and $\cQ$,
every $n\ge1$ and $\eps>0$, construct normalized finite transcript laws and
dataset-independent decoders, prove pure DP for every decoded event, and
establish the displayed expected-error bounds for every dataset.  The
real-order \cref{prop:transcript-renyi,lem:renyi-max} are included in the paper
for context but are not part of this proof chain: both the analytic proof and
Lean use integer moments, including
\cref{lem:integer-likelihood-moment} and \eqref{eq:l2-integer-moment}, and the
direct maximum-and-Hölder calculation in \eqref{eq:near-max}.

The artifact contains exact reproduction commands, a paper-to-Lean
crosswalk, and \texttt{\#print axioms} checks.  Both audited decoded theorems
report only \texttt{propext}, \texttt{Classical.choice}, and
\texttt{Quot.sound} from Lean's standard logical foundation.
\ifpublicpaper
\par\smallskip\noindent
The archive
\texttt{query-release-lean-\allowbreak supplement.tar.gz} is distributed
as ancillary material with
\href{https://arxiv.org/abs/2607.20418}{arXiv:2607.20418}; its SHA-256 digest is
\texttt{e44ad89bb0faf60d\allowbreak 31ab89ff6f4a5fe2\allowbreak
ae10cc4c686e2398\allowbreak bc8e53bd98cea6bb}.
\fi

\section{Conclusion}

We prove the remaining conjectured pure-DP query-release upper bound posed by
Nikolov and Ullman, and recover their previously resolved normalized
Euclidean rate through the same framework.  The transcript envelope converts the useful local behavior of
two PMW transcripts into globally pure-DP laws without paying either the
number of datasets or an additional $1/\eps$ factor.  The proof analyzes each
envelope using randomly sampled row changes, uses likelihood-ratio moments
only for nearby databases, switches to a mixture bound farther away, and sums
blocks rather than individual Hamming radii.  The resulting expected rates are
\[
  O\!\left(\min\left\{1,
  \sqrt{\frac{\log(2T)\log(2k)}{\eps n}}\right\}\right)
  \quad\text{and}\quad
  O\!\left(\min\left\{1,
  \sqrt{\frac{\log(2T)}{\eps n}}\right\}\right)
\]
for expected worst-coordinate and normalized Euclidean error, respectively.

\paragraph{Tools used.}
OpenAI Codex and Harmonic Aristotle were used for assistance with Lean
formalization and proof search; OpenAI Codex was also used for editorial
suggestions on clarity and organization.  The author reviewed the paper and
formalization and takes responsibility for their contents.

\appendix
\section{Supporting analytic lemmas}
\label{app:elementary-proofs}

The following proofs support the finite-space estimates invoked in the main
likelihood-envelope argument.

\subsection{Bounded likelihood ratios}
\label{app:proof-bounded-lr}

\begin{lemma}[Bounded likelihood ratio to Rényi divergence]
\label{lem:bounded-lr}
Suppose $P,Q$ have common support and
\[
  e^{-a}\le \frac{P(\omega)}{Q(\omega)}\le e^a
  \qquad\text{for every }\omega.
\]
Then, for every $r>1$,
\[
  \Ren_r(P\Vert Q)\le\frac12r a^2.
  \tag{A.1}\label{eq:bounded-lr-renyi}
\]
The same estimate composes adaptively: if each of $J$ conditional kernels has
the displayed likelihood-ratio bound with the same $a$, then the transcript
laws obey $\Ren_r(P\Vert Q)\le Jr a^2/2$.
\end{lemma}

\begin{proof}[Proof of \cref{lem:bounded-lr}]
Under $P$, put $Z=\log(P/Q)$.  Thus $Z\in[-a,a]$.  The sharp two-point
extremal bound for a bounded likelihood ratio gives
\[
  \E_P Z=\KL(P\Vert Q)
  \le a\frac{e^a-1}{e^a+1}
  =a\tanh(a/2)\le a^2/2.
\]
Indeed, under $Q$ write $R=P/Q$, so $\E_QR=1$ and
$\KL(P\Vert Q)=\E_Q[R\log R]$.  Replacing each value of
$R\in[e^{-a},e^a]$ by a random endpoint value with the same conditional mean
preserves $\E_QR$ and, by convexity of $u\mapsto u\log u$, can only increase
the KL divergence; the resulting two-point law gives the displayed bound.
Hoeffding's lemma
\cite{hoeffding1963probability} for $Z\in[-a,a]$ now gives, with $t=r-1$,
\[
  \log\E_P e^{tZ}
  \le t\E_PZ+\frac{t^2a^2}{2}
  \le\frac{t(1+t)a^2}{2}.
\]
Since $\Ren_r(P\Vert Q)=t^{-1}\log\E_Pe^{tZ}$, this proves
\eqref{eq:bounded-lr-renyi}.  For adaptive composition, factor the transcript
likelihood ratio into its conditional ratios, apply the corresponding
conditional moment estimate at the last round, and iterate the tower
property.  The logarithmic moment costs add.
\end{proof}

\subsection{Hedge regret}
\label{app:proof-hedge}

\begin{lemma}[Hedge]
\label{lem:hedge}
Let $\mu_0$ be uniform on $\cD$.  Given functions
$s_t:\cD\to[-1,1]$, define $\mu_t$ recursively by
\eqref{eq:hedge-update}, where $0<\gamma\le1$.
Then every probability distribution $\mu_*$ on $\cD$ satisfies
\[
  \sum_{t=1}^J\bigl(s_t(\mu_*)-s_t(\mu_{t-1})\bigr)
  \le\frac{\log T}{\gamma}+\frac{\gamma J}{2}.
  \tag{A.2}\label{eq:hedge-regret}
\]
\end{lemma}

\begin{proof}[Proof of \cref{lem:hedge}]
The change in relative entropy is
\[
 \KL(\mu_*\Vert \mu_t)-\KL(\mu_*\Vert \mu_{t-1})
 =\log\E_{\mu_{t-1}}e^{\gamma s_t}-\gamma s_t(\mu_*).
\]
Hoeffding's lemma bounds the logarithm by
$\gamma s_t(\mu_{t-1})+\gamma^2/2$.  Sum over $t$, use nonnegativity of the
final relative entropy, and note that
$\KL(\mu_*\Vert \mu_0)\le\log T$.  This is the multiplicative-weights
potential argument \cite{freund1997decision}.
\end{proof}

\subsection{Softmax selection gap}
\label{app:proof-softmax}

\begin{lemma}[Softmax gap]
\label{lem:softmax-gap}
Let $z_1,\ldots,z_K\in\R$, and sample $I$ with probability proportional to
$e^{\eta z_i}$, where $\eta>0$.  For
$G=\max_i z_i-z_I$,
\[
  \Prb(G\ge u)\le K e^{-\eta u},
  \qquad
  \norm{G}_{L_2}\le\frac{4\log(2K)}{\eta}.
  \tag{A.3}\label{eq:softmax-gap}
\]
\end{lemma}

\begin{proof}[Proof of \cref{lem:softmax-gap}]
The total softmax weight of indices at least $u$ below the maximum is at most
$Ke^{-\eta u}$ times the weight of a maximizer.  This proves the tail.  With
$u_0=(\log K)/\eta$ and the tail-integration identity
$\E G^2=\int_0^\infty2u\Prb(G\ge u)\,du$,
\[
  \E G^2
  \le u_0^2+2K\int_{u_0}^\infty ue^{-\eta u}\,du
  =\frac{(\log K)^2+2\log K+2}{\eta^2}.
\]
The square root is at most $4\log(2K)/\eta$.
\end{proof}

\subsection{Rényi maximum}
\label{app:proof-renyi-max}

\begin{lemma}[Rényi maximum]
\label{lem:renyi-max}
Let $N\ge1$, and let $P_0,P_1,\ldots,P_N$ have common finite support and masses
$p_0,p_1,\ldots,p_N$.  Suppose, for some $a>1$ and $A\ge0$,
\[
  \Ren_a(P_i\Vert P_0)\le aA\qquad(1\le i\le N).
  \tag{A.4}\label{eq:renyi-hyp}
\]
For every nonnegative $G$, with $b=a/(a-1)$,
\[
  \sum_\omega\max_{1\le i\le N}p_i(\omega)G(\omega)
  \le
  \exp\!\left(\frac{\log N}{a}+(a-1)A\right)
  \norm{G}_{L_b(P_0)}.
  \tag{A.5}\label{eq:renyi-max}
\]
\end{lemma}

\begin{proof}[Proof of \cref{lem:renyi-max}]
Set $R_i=p_i/p_0$.  Pointwise,
$(\max_iR_i)^a\le\sum_iR_i^a$, while
\[
  \E_{P_0}R_i^a
  =\exp\bigl((a-1)\Ren_a(P_i\Vert P_0)\bigr)
  \le e^{a(a-1)A}.
\]
Thus $\norm{\max_iR_i}_{L_a(P_0)}
\le N^{1/a}e^{(a-1)A}$.  Hölder's inequality gives
\eqref{eq:renyi-max}.
\end{proof}

\subsection{Finite blocked-radius summation}
\label{app:proof-blocked-shell-sum}

\begin{proof}[Proof of \cref{lem:blocked-shell-sum}]
The exponent in \eqref{eq:complete-shell-factor} is
\[
  -\frac\eps2-\frac{j\theta}{2}
  +\frac{\theta(j+1)}6.
\]
If $\eps\le3$, use $\theta\le6$ to bound this by
$1-j\theta/3$.  If $\eps>3$, use $\theta=\eps$ to rewrite it as
$-\eps/3-j\theta/3$.  These are exactly the two cases represented by
$u_\eps$, with one common factor $e$.  This proves
\eqref{eq:pointwise-shell-sum}.

Since $\theta\ge3$,
\[
  e^{-\theta/3}\le e^{-1}<\frac12.
\]
The finite geometric estimates
\[
  \sum_{j=0}^{N-1}2^{-j}\le2,
  \qquad
  \sum_{j=0}^{N-1}(j+1)2^{-j}\le4
\]
therefore prove \eqref{eq:geometric-sums}.
\end{proof}

\bibliographystyle{alpha}
\bibliography{references}

\end{document}